%% file: main.tex
\documentclass[sigconf,screen,nonacm,balance]{acmart}

\usepackage{xstring}
\usepackage{tabularx}

\usepackage{xr}
\startPage{1}

\usepackage{microtype} % Required by ACM guidelines

\usepackage{PaperStyle}
\usepackage{Macros}
\usepackage{bnf}

\input{acronyms}

\usepackage[capitalize]{cleveref}
\usepackage{hyperref}

\usepackage{tabularx}
\usepackage{multirow}
\usepackage{xcolor}
\PassOptionsToPackage{table}{xcolor}
\definecolor{lightgray}{gray}{0.95}
\definecolor{headergray}{gray}{0.8}
\usepackage[group-separator={,},group-minimum-digits=4]{siunitx}

\usepackage{enumitem}
\usepackage{pifont}

\crefname{rule}{Rule}{Rules}
\creflabelformat{rule}{(#2#1#3)}

\begin{document}
\title{Electronic Appendix to ``Input Invariants''}

\author{Dominic Steinh\"ofel}
\orcid{0000-0003-4439-7129}               %% \orcid is optional
\affiliation{%
  % \department{Department1}              %% \department is recommended
  \institution{CISPA Helmholtz Center for Information Security}
  \streetaddress{Stuhlsatzenhaus 5}
  \city{Saarbr\"ucken}
  \postcode{66123}
  \country{Germany}
}
\email{dominic.steinhoefel@cispa.de}

\author{Andreas Zeller}
\orcid{0000-0003-4719-8803}               %% \orcid is optional
\affiliation{%
  \institution{CISPA Helmholtz Center for Information Security}
  \streetaddress{Stuhlsatzenhaus 5}
  \city{Saarbr\"ucken}
  \postcode{66123}
  \country{Germany}
}
\email{zeller@cispa.de}

\maketitle

In this electronic appendix to our paper ``Input
Invariants''~\cite{steinhoefel.zeller-22*1}, we provide additional examples,
formal definitions, theorems, and proof sketches. Furthermore, we show the
invariants that \islearn{} mined in our evaluation (RQ3). For more information
on the \isla{} language, we also refer to the \isla{} language
specification~\cite{steinhoefel-22}.

\setcounter{section}{\getrefnumber{main-sec:isla_by_example}}
\addtocounter{section}{-1}
\DeclareRobustCommand{\sectitle}{\nameref*{main-sec:isla_by_example}}
\section{\sectitle}

In \cref{main-sec:isla_by_example}, we discussed the semantic properties
\emph{def-use} and \emph{redefinition} along the XML language. Apart from those,
there are two other re-occurring \emph{generic} constraints we would like to
discuss: \emph{Length} properties and complex conditions for which we need
dedicated \emph{semantic predicates}.

One of the target languages in our evaluation (\cref{main-sec:evaluation}) is
\acf{rest}\vphantom{\gls{rest}}, a plaintext markup language used, e.g., by
Python's docutils.  In \gls{rest}, document
(sub)titles are underlined with ``\cd{=}'' or ``\cd{-}'' symbols. However,
titles are only valid if the length of the underline is not smaller than the
length of the title text. This property cannot be expressed in a \gls{cfg};
however, we can easily capture it in an \isla{} constraint:

\begin{lstlisting}[%
  language=isla,%
  basicstyle=\footnotesize\ttfamily,%
  % gobble=2,%
  % float=htb,%
  % caption={},%
  % label=lst:,%
  % numbers=left,%
]
str.len(<section-title>.<underline>) >=
  str.len(<section-title>.<title-text>)
\end{lstlisting}

The corresponding \coreisla{} constraint is

\begin{lstlisting}[%
  language=isla,%
  basicstyle=\footnotesize\ttfamily,%
  % gobble=2,%
  % float=htb,%
  % caption={},%
  % label=lst:,%
  % numbers=left,%
]
forall <section-title> title=
    "{<title-txt> titletxt}\n{<underline> underline}" in start:
  (>= (str.len underline) (str.len titletxt))
\end{lstlisting}

There are properties which cannot be expressed using structural predicates and
SMT-LIB formulas alone. A stereotypical case are checksums occurring in many
binary formats, such as in the TAR archive file format from our benchmark set.
To account for such situations, we can extend the \isla{} language with
additional atomic assertions, so-called \emph{semantic predicates}. In contrast
to structural predicates such as \cd{inside} or \cd{same\_position}, which we
have seen before, semantic predicates do not always evaluate to false for
invalid arguments. Instead, they can suggest a \emph{satisfying solution}. The
solver logic for individual semantic predicates is implemented in Python code in
our prototype. Once this logic has been implemented, we can pass such
predicates as additional signature elements to both the \isla{} evaluator or
solver and use them in constraints. The following constraint, which is part of
our constraint set for TAR files, uses a semantic predicate \cd{tar\_checksum}
computing a correct checksum value for the header of a TAR file.

\begin{lstlisting}[%
  language=isla,%
  basicstyle=\footnotesize\ttfamily,%
  % gobble=2,%
  % float=htb,%
  % caption={},%
  % label=lst:,%
  % numbers=left,%
]
tar_checksum(<header>, <header..<checksum>)
\end{lstlisting}

This corresponds to the \coreisla{} constraint

\begin{lstlisting}[%
  language=isla,%
  basicstyle=\footnotesize\ttfamily,%
  % gobble=2,%
  % float=htb,%
  % caption={},%
  % label=lst:,%
  % numbers=left,%
]
forall <header> header in start:
  forall <checksum> checksum in header:  
    tar_checksum(header, checksum)
\end{lstlisting}

Another use case for semantic predicates is when the SMT solver frequently times
out when looking for satisfying assignments. This happens in particular for
constraints involving a complex combination of arithmetic and string (e.g.,
regular expression) constraints. For example, valid CSV files have the property
that all rows have the same numbers of columns. Assuming that we know the number
of columns in the file header, we could create a regular expression matching all
CSV lines with the same number of columns. However, if we admit quoted
expressions and a wide character range for contained text, these regular
expressions get quite complex, and the problem exceeds the capabilities of
current SMT solvers in our experience. Thus, we implemented a new semantic
predicate \cd{count} which counts the number of occurrences of some nonterminal
in an input tree, and fixes trees with an insufficient number of occurrences if
possible.
The following \isla{} constraint for the CSV property uses an additional
language feature: It introduces a \emph{numeric constant} \cd{colno} using the
\plkeyw{num} directive, which works similarly to let expressions in functional
programming languages. It is primarily---and also in this example---used to
enable information exchange between semantic predicate formulas.

\begin{lstlisting}[%
  language=isla,%
  basicstyle=\footnotesize\ttfamily,%
  % gobble=2,%
  % float=tb,%
  % caption={},%
  % label=lst:,%
  % numbers=left,%
]
forall <csv-header> hline:
  exists int colno: (
    str.to_int(colno) >= 3 and 
    str.to_int(colno) @\cd{<}@= 5 and 
    count(hline, "<raw-field>", colno) and 
    forall <csv-record> line in start:
      count(line, "<raw-field>", colno))
\end{lstlisting}

One has to be aware that the \emph{order} of semantic predicates in a constraint
matters.  This is in contrast to all other language atoms: SMT formulas, in
particular, are fed to an SMT solver only after all universal quantifiers have
been eliminated resp.\ matched, and evaluated \emph{en bloc}. Semantic
predicates, on the other hand, are generally not compositional. When computing
the checksum for a TAR file, for instance, it is important that all elements of
the file header are already fixed at that point, i.e, all semantic predicates on
header elements have to be evaluated before.  Consequently, they have to occur
before the checksum predicate in the overall constraint. Despite this
particularity, semantic predicates are an easy way to increase both the
expressiveness and solving performance of \isla{} constraints, and to overcome
the limits of SMT-LIB and off-the-shelf solvers.

\setcounter{section}{\getrefnumber{main-sec:syntax_and_semantics}}
\addtocounter{section}{-1}
\DeclareRobustCommand{\sectitle}{\nameref*{main-sec:syntax_and_semantics}}
\section{\sectitle}

We provide a more formal definition of derivation trees. We use the symbols \(<\) and
\(\leq\) to denote the strict and non-strict versions of the same partial order
relation, respectively; for the corresponding covering relation which only holds
between parents and their \emph{immediate} children, we write \(\prec\).

\begin{definition}[Derivation Tree]
  A \emph{derivation tree} for a \gls{cfg} \(G=(N,T,P,S)\) is a \emph{rooted
  ordered tree} \(t=(X,\leq_V,\leq_S)\) such that
  \begin{enumerate*}
    \item the vertices \(v\in{{X}}\) are \emph{labeled} with symbols
      \(\nlabel{v}\in{}N\cup{}T\),
    \item the \emph{vertical} order \(\leq_V\,\subseteq{}X\times{}X\) indicates
      the parent-child relation such that the partial order \((X,\leq_V)\) forms
      an unordered tree,
    \item the \emph{sibling} order \(\leq_S\,\subseteq{}X\times{}X\) yields a
      partial order \((X,\leq_S)\) such that two distinct nodes \(v_1\), \(v_2\)
      are comparable by relation \(\prec_S\) if, and only if, they are siblings,
    \item the root of \(t\) is labeled with \(S\), and\label{item:dt_root}
    \item each inner node \(v\) is labeled by a symbol in \(n\in{}N\) and, if
      \(v_1,\dots,v_k\) is the ordered list of all immediate children of \(v\),
      i.e., all distinct nodes such that \(v\prec_V{}v_i\) and \(v_i<_S{}v_l\)
      for \(1\leq{}i<l\leq{}k\), there is a production
      \((n,s_1,\dots,s_k)\in{}P\) such that \(\nlabel{v}=n\) and, for all
      \(v_i\), \(\nlabel{v_i}=s_i\).\label{item:dt_deriv}
  \end{enumerate*}
  We write \(\leaves{t}\) for the set of leaves of \(t\), and \(\nlabel{t}\) for
  the label of its root.  A derivation tree is \emph{closed} if \(l\in{}T\) for
  all \(l\in\leaves{t}\), and \emph{open} otherwise. We define 
  \(\tclosed{t}\coloneqq\forall{}l\in\leaves{t}:l\in{}T\), and
  \(\topen{t}\coloneqq\neg\tclosed{t}\).
  \(\trees{G}\) is the set of all (closed and open) derivation trees for~\(G\).
  \label{def:derivation_tree_formal}
\end{definition}

\begin{example}%
  \label{ex:xml_derivation_tree_formal}
  We explain the formal definition of derivation trees
  (\cref{def:derivation_tree_formal}) along the XML document 
  \lstinline[language=xml]!<a>x</a>!.
  visualized in Figure~\ref{main-fig:example_xml_derivation_tree} in our paper~\cite{steinhoefel.zeller-22*1}.
  Formally, this tree is represented as a triple \(t=(X,\leq_V,\leq_S)\), where
  \(X=\{v_1,\dots,v_{14}\}\), with \(\nlabel{v_1}=\bnfn{xml\text{-}tree}\),
  \(\nlabel{v_2}=\bnfn{xml\text{-}open\text{-}tag}\), etc. The \emph{vertical
  order} \(\leq_V\) contains the edges in the figure: For example,
  \(v_1\leq_V{}v_2\) and \(v_2\leq_V{}v_{13}\). This relation alone only gives
  us an unordered tree: When ``unparsing'' the tree, we could thus obtain the
  undesired result
  \lstinline[language=xml]!xa><></a!.
  Thus, we define a \emph{sibling order} \(\leq_S\) to order the immediate
  children of the same node. For instance, we have \(v_6\leq_S{}v_8\) (and
  \(v_6\leq_S{}v_6\), \(v_6<_S{}v_8\), and \(v_6\prec_S{}v_7\)), but not
  \(v_6\leq_S{}v_9\) (since they have different parents) and \(v_6\prec_S{}v_8\)
  (since they are not \emph{immediate} siblings).
  The tree is not only any ordered tree, but a \emph{derivation tree} for the
  XML grammar in Figure~\ref{main-fig:xml-bnf} in \cite{steinhoefel.zeller-22*1}, since it satisfies
  4 and 5 of Definition 3.4 in \cite{steinhoefel.zeller-22*1}. The root of
  the tree, \(v_1\), is labeled with the grammar's start symbol
  \(\bnfn{xml\text{-}tree}\) (\cref{item:dt_root}). The tree relations conform to
  the possible grammar derivations: Consider, e.g., node \(v_2\) and its
  immediate children \(v_6\), \(v_7\), and \(v_8\). According to
  \cref{item:dt_deriv}, there has to be a production
  \((\bnfn{xml\text{-}open\text{-}tag},\bnft{<},\bnfn{id},\bnft{>})\) in the
  grammar, which is indeed the case, since 
  \(\bnft{<} \bnfs \bnfn{id} \bnfs \bnft{>}\) 
  is an expansion alternative (the first one) for the nonterminal
  \(\bnfn{xml\text{-}open\text{-}tag}\). The leaf set \(\leaves{t}\) is
  \(\{v_6,v_{13},v_8,v_9,v_{10},v_{14},v_{12}\}\). The tree \(t\) is
  \emph{closed}, since all leaves are labeled with \emph{terminal} symbols. It
  would be \emph{open} if we removed the subtree rooted in any tree node (but
  the root).
\end{example}

\paragraph{Standard \isla Predicates}%
\label{par:standard_isla_predicates}

\isla{} offers a catalog of default supported predicates.
\Cref{tab:standard_isla_predicates} provides an overview of those. Structural
predicates can be re-used for many different languages, while semantic
predicates are mostly application-specific. For this reason, there is only one
semantic predicate included in \isla{} per default, which is the \cd{count}
predicate used in the formalization of CSV.

\begin{table*}[tb]
  \centering
  \caption{Standard \isla{} predicates. The predicate \cd{count} is a semantic predicate; all
  other predicates are structural predicates.}
  \label{tab:standard_isla_predicates}
  \begin{tabularx}{\linewidth}{lX}
   \toprule
   Predicate & Explanation \\\midrule
   \cd{after($\mathit{node}_1$, $\mathit{node}_2$)} & \cd{$\mathit{node}_1$} occurs after \cd{$\mathit{node}_2$} (not below) in the parse tree. \\
   \cd{before($\mathit{node}_1$, $\mathit{node}_2$)} & \cd{$\mathit{node}_1$} occurs before \cd{$\mathit{node}_2$} (not below) in the parse tree. \\
   \cd{consecutive($\mathit{node}_1$, $\mathit{node}_2$)} & \cd{$\mathit{node}_1$} and \cd{$\mathit{node}_2$} are consecutive leaves in the parse tree. \\
   \cd{different\_position($\mathit{node}_1$, $\mathit{node}_2$)} & \cd{$\mathit{node}_1$} and \cd{$\mathit{node}_2$} occur at different positions (cannot be the same node). \\
   \cd{direct\_child($\mathit{node}_1$, $\mathit{node}_2$)} & \cd{$\mathit{node}_1$} is a direct child of \cd{$\mathit{node}_2$} in the derivation tree. \\
   \cd{inside($\mathit{node}_1$, $\mathit{node}_2$)} & \cd{$\mathit{node}_1$} is a subtree of \cd{$\mathit{node}_2$.} \\
   \cd{level(PRED, NONTERMINAL, $\mathit{node}_1$, $\mathit{node}_2$)} & \cd{$\mathit{node}_1$} and \cd{$\mathit{node}_2$} are related relatively to each other as specified by \cd{PRED} and \cd{NONTERMINAL} (see below). \cd{PRED} and \cd{NONTERMINAL} are strings. \\
   \cd{nth(N, $\mathit{node}_1$, $\mathit{node}_2$)} & \cd{$\mathit{node}_1$} is the N-th occurrence of a node with its nonterminal symbol within \cd{$\mathit{node}_2$.} N is a numeric String. \\
   \cd{same\_position($\mathit{node}_1$, $\mathit{node}_2$)} &
   \cd{$\mathit{node}_1$} and \cd{$\mathit{node}_2$} occur at the same position (have to be the same node). \\\midrule
   \cd{count(in\_tree, NEEDLE, NUM)} & There are \cd{NUM} occurrences of the \cd{NEEDLE} nonterminal in \cd{in\_tree}. \cd{NEEDLE} is a string, \cd{NUM} a numeric string or int variable. \\\bottomrule
  \end{tabularx}
\end{table*}

\paragraph{Matching Match Expressions}%
\label{par:matching_match_expressions}

Match expressions are matched against derivation trees by first parsing them
into \emph{abstract} parse trees (with open leaves), and then matching these
parse trees against the derivation tree in question. This process is also
described in detail in the \isla{} language specification~\cite{steinhoefel-22}.

We use a function \(\mathit{mexprTrees}(T,\mathit{mexpr})\) that takes a
nonterminal \(T\) and a match expression \(\mathit{mexpr}\) and returns a set of
derivation trees. If the match expression contains \emph{optional} elements, it
is ``flattened'' first. That is, we compute all combinations of activated and
non-activated optional expressions. If there are \(n\) optionals in the match
expression, we obtain \(2^n\) flattened match expressions. Then, we parse the
flattened expressions using an augmented version of the reference grammar. The
augmented grammar adds expansions ``\cd{<A> ::= '<' 'A' '>'}'' for each
nonterminal \cd{A}, and similarly extends the grammar with expansions for
variable binders ``\cd{\{<T> var\}}.'' Due to ambiguities in the grammar, we
might obtain multiple parse trees even for flattened expressions;
\(\mathit{mexprTrees}\) returns all of them, along with a mapping of bound
variables to the positions of their matches in the respective derivation trees.
After parsing the match expression, the function \(\mathit{matchTrees}(t,t',P)\)
matches a derivation tree \(t\) against a result \((t',P)\) from
\(\mathit{mexprTrees}\), where \(t'\) is a parse tree and \(P\) a mapping from
bound variables to positions in \(t'\). \Cref{fig:defn_match_trees} shows the
definition of \(\mathit{mexprTrees}\). In the definition,

\begin{itemize}
  \item \(l(t)\) is the label of the tree \(t\);
  \item all alternatives in the definition are *mutually exclusive* (the first
    applicable one is applied);
  \item by \(\mathit{numc}(t)\) we denote the number of  children of the
    derivation tree \(t\);
  \item by \(\mathit{child}(t, i)\) we denote the \(i\)-th child of t, starting
    with 1; 
  \item \(P_i\) is computed from a mapping \(P\) by discarding all paths in
    \(P\) not starting with \(i\) and taking the \emph{tail} (discarding the
    first element) for all other paths; and
  \item we use standard set union notation \(\bigcup_{i=1}^n\beta_i\) for
    combining variable assignments \(\beta_i\).
\end{itemize}

Let \(T\) be the label of the root of tree \(t\). We define
\begin{multline*}
  \mathit{match}(t,\mathit{mexpr})\coloneqq\\
  \left(\bigcup_{(t',P)\in\mathit{mexprTrees}(T,\mathit{mexpr})}
  \{\mathit{matchTrees}(t,t',P)\}\right)
  \setminus\{\bot\}
\end{multline*}

\begin{figure*}[tb]
  \centering
  \(
  \displaystyle
  \mathit{matchTrees}(t, t', P) :=
  \begin{cases}
  \bot                                                                                  & \text{if }l(t)\neq{}l(t')\vee(\mathit{numc}(t')>0\wedge \\
                                                                                        & \qquad\mathit{numc}(t)\neq\mathit{numc}(t')) \\
  [v\mapsto{}t]                                                                         & \text{if }P=[v\mapsto{}()]\text{ for some }v \\
  \bot                                                                                  & \text{if }\mathit{matchTrees}(\mathit{child}(t, i), \mathit{child}(t', i), P_i)=\bot \\
                                                                                        & \qquad\text{for any }i\in[1,\dots,\mathit{numc}(t)] \\
  \bigcup_{i=1}^{\mathit{numc}(t)}\Big( & \\
  \quad\mathit{matchTrees}\big(\mathit{child}(t, i),                  & \\
  \quad\hphantom{\mathit{match}\big(}\mathit{child}(t', i), P_i\big)\Big) & \text{otherwise} \\
  \end{cases}
  \)
  \caption{Recursive Definition of \(matchTrees\).}%
  \label{fig:defn_match_trees}
\end{figure*}

\setcounter{section}{\getrefnumber{main-sec:input_generation}}
\addtocounter{section}{-1}
\DeclareRobustCommand{\sectitle}{\nameref*{main-sec:input_generation}}
\section{\sectitle}

We provide a formalization of our \isla{} constraint solver, including two
correctness theorems and proof sketches.

We formalize input generation for \isla{} as a transition system between
\emph{\glspl{cdt}} \(\condtree{\Phi}{t}\), where \(\Phi\) is a set of \isla{}
formulas (interpreted as a conjunction) and \(t\) a (possibly open) derivation
tree. Intuitively, \(\Land\Phi\) constrains the inputs represented by \(t\), similarly as \(\forsem{\phi}\) constrains the language of the grammar. To
make this possible, we need to relax the definition of \isla{} formulas: Instead
of free variables, formulas may contain references to tree nodes which they are
concerned about. To that end, tree nodes are assigned unique, numeric
identifiers, which may occur everywhere in \isla{} formulas where a free
variable might occur (variables bound by quantifiers may not be replaced with
tree identifiers).

Consider, for example, the \isla{} constraint
\[\phi=\islaforall[\bnfn{id}]{\mathit{id}}{\mathit{start}}{\cd{(= (str.len $\mathit{id}$) 17)}}\]
constraining the XML grammar in Figure~\ref{main-fig:xml-bnf} in \cite{steinhoefel.zeller-22*1} to identifiers of length 17,
where \(\mathit{id}\) is a bound variable of type \(\bnfn{ID}\) and
\(\mathit{start}\) is a free variable of type \(\bnfn{start}\). Let \(t\) be a
tree consisting of a single (root) node with identifier 1, and labeled with
\(\bnfn{start}\). Then, \(\forsem{\phi}\) is identical to the strings
represented by the \gls{cdt}
\[\condtree{\{\islaforall[\bnfn{id}]{\mathit{id}}{\mathit{1}}{\cd{(=~(str.len~$\mathit{id}$)~17)}}\}}{t}.\]
Our \gls{cdt} transition system relates an input \gls{cdt} to a set of output
\glspl{cdt}. We define two properties of such transitions: A transition is
\emph{precise} if the input represents \emph{at most} the set of all strings
represented by all outputs together; conversely, it is \emph{complete} if the
input represents \emph{at least} the set of all strings represented by all
outputs. Precision is mandatory for the \isla{} producer, since we have to avoid
generating system inputs which do not satisfy the specified constraints. 

To define the semantics of \glspl{cdt}, we first define the closed trees
represented by (the language of) \emph{open} derivation trees. We need the
concept of a \emph{tree substitution}: The tree \(t\update{v}{t'}\) results from
\(t=(X,\leq_V,\leq_S)\) by replacing the subtree rooted in node \(v\in{}X\) by
\(t'\), updating \(X\), \(\leq_V\) and \(\leq_S\) accordingly.

\begin{definition}[Semantics of Open Derivation Trees]
  Let \(t\in\trees{G}\) be a derivation tree for a grammar \(G=(N,T,P,S)\). We
  define the set \(\trees{t}\subseteq\trees{G}\) of closed derivation trees
  represented by \(t\) as
  \[
    \begin{alignedat}{3}
      \trees{t}      & \coloneqq\, &  & \big\{t\update{l_1}{t_1}\cdots\update{l_k}{t_k}\;\vert \\
                     &             &  & \quad{}\phantom{{\land\;}}l_i\in\leaves{t}\land{}k=\card{\leaves{t}} \\
                     &             &  & \quad{}{\land\;}(\forall{}j,m\in1\dots{}k:l\neq{}m\limp{}l_j\neq{}l_m) \\
                     &             &  & \quad{}{\land\;}n_i=\nlabel{t_i}=\nlabel{l_i} \\
                     &             &  & \quad{}{\land\;}G_{n_i}=(N,T,P,n_i)\land{}t_i\in\trees{G_{n_i}} \big\}
    \end{alignedat}
  \]
  \label{def:semantics_of_open_derivation_trees}
\end{definition}

Observe that for the tree consisting of a single node labeled with the start
symbol \(S\), \(\trees{t}\) is identical to \(\trees{S}\). Furthermore, for any
\emph{closed} tree \(t'\), it holds that \(\trees{t'}=\{t'\}\).

We re-use the validity judgment defined from Definition~\ref{main-def:isla_validation} in our paper
\cite{steinhoefel.zeller-22*1} for the semantics definition for \glspl{cdt} by
interpreting tree identifiers in formulas similarly to variables. Furthermore,
the special variable assignment \(\varassgn_t\) for the derivation tree \(t\)
associates with each tree identifier in \(t\) the subtree rooted in the node
with that identifier.  Then, the definition is a straightforward specialization
of Definition~3.8 from \cite{steinhoefel.zeller-22*1}:

\begin{definition}[Semantics of \glspl{cdt}]
  Let \(\Phi\subseteq\ps{\fml}\) be a set of \isla{} formulas for the signature
  \(\sig=(G,\psym,\vsym)\), \(t\in\trees{G}\) be a derivation tree, and
  \(\pinterpret\), \(\sinterpret\) be interpretations for predicates and SMT
  formulas.
  We define the semantics \(\cdtsem{\condtree{\Phi}{t}}\) of the \gls{cdt}
  \(\condtree{\Phi}{t}\) as
  \[
    \cdtsem{\condtree{\Phi}{t}}\coloneqq\{\treetostr{t'}\;\vert\;t'\in\trees{t}\land
      \tclosed{t'} 
    \land\bislam{\varassgn_{t'}}{\Land\Phi}\}.
  \]
  \label{def:semantics_of_cdts}
\end{definition}

A \acf{cdtts}\vphantom{\gls{cdtts}} is simply a transition system between
\glspl{cdt}.

\begin{definition}[\gls{cdt} Transition System]
  A \emph{\gls{cdtts}} for a signature \(\sig=(G,\psym,\vsym)\) is a transition
  system \mbox{\((C,\transrel)\)}, where, for \(\Phi\in\ps{\fml}\) and
  \(t\in\trees{G}\), \(C\) consists of \glspl{cdt} \(\condtree{\Phi}{t}\), and
  \(\transrel\,\subseteq{}C\times{}C\). We write
  \(\mathit{cdt}\transrel\mathit{cdt'}\) if
  \((\mathit{cdt},\mathit{cdt'})\in'\,\transrel\).
  \label{def:cdt_transition_system}
\end{definition}

%\begin{figure}[tb]
%  \centering
%  \begin{alignat}{2}
%    & \condtree{\Phi}{t} \; \transrel \; \{\condtree{\Phi}{t'}\;\vert\;t'\in\trees{t}\}\neq\{\condtree{\Phi}{t}\}  \label[rule]{rule:isla_naive_enum}\\
%    & \condtree{\Phi}{t} \; \transrel \; \{\condtree{\emptyset}{t}\} \text{ if }\bislam{\varassgn_t}{\Land\Phi} \label[rule]{rule:isla_naive_simpl}\\
%    & \condtree{\Phi}{t} \; \transrel \; \emptyset \text{ if }\notbislam{\varassgn_t}{\Land\Phi} \label[rule]{rule:isla_naive_discard}
%  \end{alignat}
%  \caption{Naive \isla \gls{cdtts} Transition Relation}%
%  \label{fig:naive_cdtts_transrel}
%\end{figure}

% \cref{fig:naive_cdtts_transrel} shows the definition of a naive
% ``generate-and-check'' \gls{cdtts} for \isla. The transition rules are ordered
% and mutually exclusive, i.e., \cref{rule:isla_naive_simpl} only applies to
% \glspl{cdt} for which \cref{rule:isla_naive_enum} does not apply since \(t\) is
% already concrete. The \gls{cdtts} enumerates all concrete trees
% (\cref{rule:isla_naive_enum}), simplifies constraints to \(\emptyset\) (which
% represents truth) if the current (concrete) tree satisfies its constraint, and
% discards the \gls{cdt} otherwise (\cref{rule:isla_naive_discard}).

Intuitively, one applies \gls{cdtts} transitions to an initial constraint with a
trivial tree only consisting of a root node labeled with the start nonterminal,
and collects ``output'' \glspl{cdt} \(\condtree{\emptyset}{t}\) with an empty
constraint. The trees \(t\) of such outputs are solutions to the initial
problem. We call a \gls{cdtts} \emph{globally precise} if all such trees \(t\)
are actual solutions, i.e., the system does not produce wrong outputs; we call
it \emph{globally complete} if the entirety of trees \(t\) from result
\glspl{cdt} represents the full semantics of the input \gls{cdt}.

\begin{definition}[Global Precision and Completeness]
  Let \((C,\transrel)\) be a \gls{cdtts}, and \(R_\mathit{cdt}\) be the set of
  all closed trees \(t\) such  that
  \(\mathit{cdt}\transrel\cdots\transrel\condtree{\emptyset}{t}\) is a
  derivation in \((C,\transrel)\).  Then, \((C,\transrel)\) is \emph{globally
  precise} if, for each \gls{cdt} \(\mathit{cdt}\) in the domain of
  \(\transrel\), it holds that
  \(\cdtsem{\mathit{cdt}}\supseteq\{\treetostr{t}\vert{}t\in{}R_\mathit{cdt}\}\).
  The \gls{cdtts} is \emph{globally complete} if it holds that
  \(\cdtsem{\mathit{cdt}}\subseteq\{\treetostr{t}\vert{}t\in{}R_\mathit{cdt}\}\).
  \label{def:global_precision_completeness_of_cdtts}
\end{definition}

To enable transition-local reasoning about precision and completeness, we define
notions of local precision and completeness.
Local precision is the property that \emph{at each transition step}, no
``wrong'' inputs are added, and local completeness the property that no
transition step loses information.

\begin{definition}[Local Precision and Completeness]
  A \gls{cdtts} \((C,\transrel)\) is \emph{precise} if, for each \gls{cdt}
  \(\mathit{cdt}\) in the domain of \(\transrel\), it holds that
  \(\cdtsem{\mathit{cdt}}\supseteq\bigcup_{\mathit{cdt}\transrel\mathit{cdt'}}(\cdtsem{\mathit{cdt'}})\).
  The \gls{cdtts} is \emph{complete} if it holds that
  \(\cdtsem{\mathit{cdt}}\subseteq\bigcup_{\mathit{cdt}\transrel\mathit{cdt'}}(\cdtsem{\mathit{cdt'}})\).
  \label{def:precision_completeness_of_cdtts}
\end{definition}

As for ``soundness'' in first-order logic (see, e.g.,~\cite{dalen-94}), local
precision implies global precision, i.e., it suffices to show that the
individual transitions are precise to obtain the property for the whole system.
This is demonstrated by the following \cref{lem:local_implies_global_precision}.
Note that the opposite direction does not hold, since a \gls{cdtts} could in
theory lose precision locally and recover it globally, although it is unclear
how (and why) such a system should be designed.

\begin{lemma}
  A locally precise \gls{cdtts} is also globally precise.
  \label{lem:local_implies_global_precision}
\end{lemma}

\begin{proof}
  The lemma trivially holds if \(R_\mathit{cdt}=\emptyset\). Otherwise, let
  \(\mathit{cdt}_0\transrel\mathit{cdt_1}\transrel\cdots\transrel\condtree{\emptyset}{t}\)
  be any transition chain s.t.~\(\mathit{cdt}_0=\mathit{cdt}\) and
  \(t\in{}R_\mathit{cdt}\).
  Then, it follows from local precision that
  \(\cdtsem{\mathit{cdt}_k}\supseteq\cdtsem{\mathit{cdt}_{k+1}}\) for
  \(k=0,\dots,n-1\), and by transitivity of \(\supseteq\) also
  \(\cdtsem{\mathit{cdt}_k}\supseteq\cdtsem{\mathit{cdt}_{l}}\) for
  \(0\leq{}k<l\leq{}n\). Since
  \(\cdtsem{\condtree{\emptyset}{t'}}=\treetostr{t'}\) for closed \(t'\), the
  lemma follows.
\end{proof}

Global completeness cannot easily be reduced to local completeness: It includes
the ``termination'' property that all derivations end in \glspl{cdt} with empty
constraint set; furthermore, one has to show that there is an applicable
transition for each \gls{cdt} with non-empty semantic. 

%The \gls{cdtts} in \cref{fig:naive_cdtts_transrel} is globally precise
%\emph{and} complete, but very \emph{inefficient} especially for sparse
%constraints.

Our \isla{} solver prototype implements the \gls{cdtts} in
\cref{fig:isla_cdtts}. It solves SMT and semantic predicate constraints by
querying the SMT solver or the predicate oracle, and eliminates existential
constraints by inserting new subtrees into the current conditioned tree. Only
when the complete constraint has been eliminated, we finish off the remaining
incomplete tree by replacing open leaves with suitable concrete subtrees.
This is in principle a complete procedure; yet, our implementation only
considers a finite subset of all solutions in solver queries and when performing
tree insertion. Consequently, it usually misses some solutions, but outputs
\emph{more diverse results more quickly} compared, e.g., to a naive search-based
approach filtering out wrong solutions.

\newcommand{\id}{\mathit{id}}
\setcounter{equation}{0}
\begin{figure*}[tb]
  % \small
  % \centering
  \begin{multicols}{2}
    \begin{alignat}{2}
      % INV
      & {\condtreeidx{\{\dots,\phi,\dots\}}{I}{t} \; \transrel \; \{}\condtreeidx{\{\dots,\phi'\dots\}}{I}{t}\;\vert\; \label[rule]{rule:inv}\\
      & \hphantom{\condtreeidx{\{\dots,\phi,\dots\}}{I}{t} \; \transrel \; \{}\quad\phi'\in\establishinv{\phi}\land\phi\neq\phi'\}\neq\emptyset \nonumber\\
      % ELIM STRUCTRUAL - FALSE
      % & \condtreeidx{\Phi}{I}{t} \; \transrel \; \emptyset\text{ if }\phi\in\Phi\cap\psymstruct\land\;\pinterpret(\phi)=\semfalse \label[rule]{rule:elim_struct_false}\\
      % ELIM STRUCTRUAL - TRUE
      & \condtreeidx{\Phi}{I}{t} \; \transrel \; \{\condtreeidx{\Phi\setminus\phi}{I}{t}\;\vert\;\phi\in\Phi\cap\psymstruct\land\;\pinterpret(\phi)=\semtrue\} \label[rule]{rule:elim_struct_true}\\
      % ELIM EXISTENTIAL NUMERIC QUANTIFIER
      & \condtreeidx{\{\dots,\islaforall[\numtype]{n}{\phi},\dots\}}{I}{t} \; \label[rule]{rule:elim_num_intro}\transrel\\
      & \quad\quad \condtreeidx{\{\dots,\singlesubst{n}{c}(\phi),\dots\}}{I}{t} \nonumber\\
      & \quad\text{where }c\in\vsym\text{ is \emph{fresh} and }\vartype{c}=\numtype \nonumber\\
      % ELIM UNIVERSAL NUMERIC QUANTIFIER (--> TODO)
      % ELIM FORALL
      & \condtreeidx{\{\dots, \overbrace{\islaforall[\mathit{type}]{v}{\id}{\phi}}^\psi, \dots\}}{I}{t} \;\transrel\; \{\condtreeidx{\{\dots,\, \dots\}}{I}{t}\} \label[rule]{rule:elim_forall}\\
      & \quad\text{if }\forall\text{ subrees }t'\text{ of }\trees{\varassgn_t(\id)}: \nonumber\\
      & \quad\quad(\psi,t')\in{}I\,\lor\nlabel{t'}\neq\mathit{type} \nonumber\\
      % ELIM FORALL-MEXPR
      & \condtreeidx{\{\dots, \overbrace{\islaforallb[\mathit{type}]{v}{\mathit{mexpr}}{\id}{\phi}}^\psi, \dots\}}{I}{t} \;\transrel \label[rule]{rule:elim_forall_mexpr}\\
      & \quad\{\condtreeidx{\{\dots,\, \dots\}}{I}{t}\}\text{ if } \forall\text{ subrees }t'\text{ of }\trees{\varassgn_t(\id)}: \nonumber\\
      & \quad\quad\quad(\psi,t')\in{}I\,\lor\nlabel{t'}\neq\mathit{type}\,\lor \nonumber\\
      & \quad\quad\quad\text{there is no }m=\matchme{t'}{\mathit{mexpr}} \nonumber\\
      % MATCH FORALL
      & \condtreeidx{\{\dots, \overbrace{\islaforall[\mathit{type}]{v}{\id}{\phi}}^{\psi}, \dots\}}{I}{t} \;\transrel \label[rule]{rule:match_forall}\\
      & \quad\big\{\condtreeidx{\{\dots,\psi,\dots\}\cup\bigcup\Psi}{I\,\cup\,(\{\psi\}\times{}T)}{t}\big\}~\mathrm{where} \nonumber\\
      & \quad\quad\Psi=\big\{\updatesubst{\beta_t}{v}{t'}(\phi) \;\vert\;t'\in{}T\big\}\,\land \nonumber\\
      & \quad\quad{}T=\big\{t'\;\vert\;t'=\varassgn_t(\id)\land(\psi,t')\notin{}I\,\land \nonumber\\
      & \quad\quad{}\hphantom{T=\big\{t'\;\vert\;}\nlabel{t'}=\mathit{type}\big\}\neq\emptyset \nonumber \\
      % MATCH FORALL-MEXPR
      & \condtreeidx{\{\dots, \overbrace{\islaforallb[\mathit{type}]{v}{\mathit{mexpr}}{\id}{\phi}}^{\psi}, \dots\}}{I}{t} \;\transrel \label[rule]{rule:match_forall_mexpr}\\
      & \quad\big\{\condtreeidx{\{\dots,\psi,\dots\}\cup\bigcup\Psi}{I\,\cup\,(\{\psi\}\times{}T)}{t}\big\}~\mathrm{where} \nonumber\\
      & \quad\quad\Psi=\big\{\updatesubstm{\updatesubst{\beta_t}{v}{t'}}{m}(\phi) \;\vert\;(t',m)\in{}T\big\}\,\land \nonumber\\
      & \quad\quad{}T=\big\{(t',m)\;\vert\;t'=\varassgn_t(\id)\land(\psi,t')\notin{}I\,\land \nonumber\\
      & \quad\quad\quad{}\nlabel{t'}=\mathit{type}\;\land \nonumber\\
      & \quad\quad\quad\text{there is an }m=\matchme{t}{\mathit{mexpr}}\big\}\neq\emptyset \nonumber %\\
    \end{alignat}
  
    \vfill
    \columnbreak

    \begin{alignat}{2}
      % EXPAND
      & \condtreeidx{\Phi}{I}{t} \; \transrel \; \{\condtreeidx{\Phi}{I}{t'}\;\vert\;t'\in\expand{\Phi}{t}\neq\emptyset\} \label[rule]{rule:expand}\\
      % ELIM SMT - UNSAT
      % & {\condtreeidx{\Phi}{I}{t} \; \transrel \; \emptyset~\mathrm{if}~}\phi\in\Phi\cap\boolterms{\vsym}\land\sinterpret(\phi)=\semfalse \label[rule]{rule:elim_smt_unsat}\\
      % ELIM SMT - SAT
      & {\condtreeidx{\Phi}{I}{t} \; \transrel \; \{}\condtreeidx{\Phi\setminus\phi}{I}{\varassgn(t)}\;\vert\; \label[rule]{rule:elim_smt_sat}\\
      & \quad\quad\phi\in\Phi\cap\boolterms{\vsym}\land\varassgn\in\sinterpret(\phi)\neq\semfalse\} \nonumber\\
      % ELIM SEMANTIC PREDICATE - UNSAT
      % & {\condtreeidx{\Phi}{I}{t} \; \transrel \; \emptyset~\mathrm{if}~}\phi\in\Phi\cap\psymsem\land\pinterpret(\phi)=\semfalse \label[rule]{rule:elim_sempred_unsat}\\
      % ELIM SEMANTIC PREDICATE - SAT
      & {\condtreeidx{\Phi}{I}{t} \; \transrel \; \{}\condtreeidx{\Phi\setminus\phi}{I}{\varassgn(t)}\;\vert\; \label[rule]{rule:elim_sempred_sat}\\
      & \quad\quad\phi\text{ is \emph{first} }\phi\in\Phi\cap\psymsem\land\varassgn\in\pinterpret(\phi)\notin\{\semfalse,\notready\}\} \nonumber\\
      % MATCH EXISTS
      & \condtreeidx{\{\dots, \overbrace{\islaexists[\mathit{type}]{v}{\id}{\phi}}^{\psi}, \dots\}}{I}{t} \;\transrel \label[rule]{rule:match_exists}\\
      & \quad\bigcup_{\xi\in\Psi}\big\{\condtreeidx{\{\dots,\xi,\dots\}}{I}{t}\big\}~\mathrm{where} \nonumber\\
      & \quad\quad\Psi=\big\{{\updatesubst{\beta_t}{v}{t'}}(\phi) \;\vert\;t'\in{}T\big\}\,\land \nonumber\\
      & \quad\quad{}T=\big\{t'\;\vert\;t'=\varassgn_t(\id)\land(\psi,t')\notin{}I\,\land \nonumber\\
      & \quad\quad\quad{}\nlabel{t'}=\mathit{type}\,\big\}\neq\emptyset \nonumber\\
      % MATCH EXISTS-MEXPR
      & \condtreeidx{\{\dots, \overbrace{\islaexistsb[\mathit{type}]{v}{\mathit{mexpr}}{\id}{\phi}}^{\psi}, \dots\}}{I}{t} \;\transrel \label[rule]{rule:match_exists_mexpr}\\
      & \quad\bigcup_{\xi\in\Psi}\big\{\condtreeidx{\{\dots,\xi,\dots\}}{I}{t}\big\}~\mathrm{where} \nonumber\\
      & \quad\quad\Psi=\big\{\updatesubstm{\updatesubst{\beta_t}{v}{t'}}{m}(\phi) \;\vert\;(t',m)\in{}T\big\}\,\land \nonumber\\
      & \quad\quad{}T=\big\{(t',m)\;\vert\;t'=\varassgn_t(\id)\land(\psi,t')\notin{}I\,\land \nonumber\\
      & \quad\quad\quad{}\nlabel{t'}=\mathit{type}\;\land \nonumber\\
      & \quad\quad\quad\text{there is an }m=\matchme{t}{\mathit{mexpr}}\big\}\neq\emptyset \nonumber\\
      % ELIM EXISTS
      & \condtreeidx{\{\dots, \islaexists[\mathit{type}]{v}{\id}{\phi}, \dots\}}{I}{t} \;\transrel \label[rule]{rule:elim_exists}\\
      & \quad\big\{\condtreeidx{\{\dots,\singlesubst{v}{\mathit{nid}}(\phi),\dots\}\cup\Phi_\mathit{orig}}{I}{\singlesubst{\id}{t'}(t)}\;\vert \nonumber\\
      & \quad\quad(\mathit{nid},t')\in\trinsert{\maketree{v}}{\varassgn_t(\id)} \nonumber\\
      % ELIM EXISTS-MEXPR
      & \condtreeidx{\{\dots, \islaexistsb[\mathit{type}]{v}{\mathit{mexpr}}{\id}{\phi}, \dots\}}{I}{t} \;\transrel \label[rule]{rule:elim_exists_mexpr}\\
      & \quad\big\{\condtreeidx{\{\dots,\singlesubst{v}{\mathit{nid}}(\phi),\dots\}\cup\Phi_\mathit{orig}}{I}{\singlesubst{\id}{t'}(t)}\;\vert \nonumber\\
      & \quad\quad(\mathit{nid},t')\in\trinsert{\maketree{v, \mathit{mexpr}}}{\varassgn_t(\id)} \nonumber\\
      % FINISH TRIVIAL CONSTRAINT
      & \condtreeidx{\emptyset}{I}{t} \; \transrel \; \{\condtreeidx{\emptyset}{I}{t'}\;\vert\;t'\in\trees{t}\}\neq\{\condtreeidx{\emptyset}{I}{t}\} \label[rule]{rule:finish_trivial}\\
      % FINISH SEMANTIC CONSTRAINTS
      & \condtreeidx{\Phi}{I}{t} \; \transrel \; \{\condtreeidx{\Phi}{I}{t'}\;\vert\;t'\in\trees{t}\}\neq\{\condtreeidx{\Phi}{I}{t}\} \label[rule]{rule:finish_semantic}\\
      & \quad\text{ if }\Phi\subseteq\{p(\dots)\;\vert\;p\in\psymsem\} \nonumber
    \end{alignat}
  \end{multicols}
  \vspace{-\baselineskip}
  \caption{Efficient \isla{} \gls{cdtts} Transition Relation}%
  \label{fig:isla_cdtts}
\end{figure*}

\paragraph{Transition Rules}

In the \isla{} \gls{cdtts}, we use \emph{indexed} \glspl{cdt}
\(\condtreeidx{\Phi}{I}{t}\). In the index set \(I\), we track previous matches
of universal quantifiers to make sure that we do not match the same trees over
and over. 
Since SMT formulas can now also contain variables, evaluating them can result in
a \emph{model} \(\varassgn\) (an assignment). Note that we can obtain different
assignments by repeated solver calls (negating previous solutions).
We divide the set \(\psym\) of predicate symbols into two disjoint sets
\(\psymstruct\) and \(\psymsem\) of \emph{structural} and \emph{semantic}
predicates.
Structural predicates address constraints such as \emph{before} or
\emph{within}, and they evaluate to \(\semtrue\) or \(\semfalse\).
\emph{Semantic} predicates formalize constraints such as specific checksum
implementations. They may additionally evaluate to a set of assignments, as in
the case of satisfiable SMT expressions, or to the special value ``not ready''
(denoted by \(\notready\)). Intuitively, an evaluation results in \(\semtrue\)
(\(\semfalse\)) if all of (not any of) the derivation trees represented by an
abstract tree satisfy the predicate. Assignments are returned if the given tree
can be completed or ``fixed'' to a satisfying solution. One may obtain
\(\notready\) if the constrained tree lacks sufficient information for such a
computation (e.g., the inputs of a checksum predicate are not yet determined).

We explain the individual transition rules of the \gls{cdtts} from
\cref{fig:isla_cdtts}.
\cref{rule:inv} uses a function \(\establishinvf:\fml\rightarrow\ps{\fml}\) to
enforce the invariant that all formulas \(\phi\in\Phi\) are in \gls{nnf} and do
not contain top-level conjunctions and disjunctions. Basically,
\(\establishinvf\) converts its input into \acl{dnf} and returns the disjunctive
elements.  It is only applicable to \glspl{cdt} whose constraints do not satisfy
the invariant.
\cref{rule:elim_struct_true} eliminates satisfied structural predicate formulas
from a constraint set.
Existential quantifiers over numbers are eliminated in
\cref{rule:elim_num_intro} by introducing a fresh (not occurring in the
containing \gls{cdt}) variable symbol with the special nonterminal type
\(\numtype\) for natural numbers.

\cref{rule:elim_forall,rule:elim_forall_mexpr} eliminate universal formulas that
have already been matched with all applicable subtrees, and which cannot
possibly be matched against \emph{any} extension of the (open) tree. This is the
case if the nonterminal type of the quantified variable is not reachable from
any leaf and, if there is a match expression, the current tree cannot be
completed to a matching one.

Universal formulas with and without match expressions are subject of
\cref{rule:match_forall,rule:match_forall_mexpr}. First, matching subtrees of
the tree \(\varassgn_t(\id)\) identified with \(\id\) are collected in a set
\(T\). We only consider subtrees that are not already matched, i.e., where
\((\psi,t')\) is not yet in the index set \(I\). If \(T\) is empty, the rules
are not applicable. Otherwise, the set \(\Phi\) of all instantiations of
\(\phi\) according to the discovered matches is added to the constraint set. We
record the instantiations \((\psi,t')\), for all matched trees \(t'\), in the
index set. The output of these rules is a singleton.

If universal quantifiers remain which cannot be matched or eliminated, we expand
the current tree in \cref{rule:expand}. The function \(\expand{\Phi}{t}\)
returns all possible trees \(t'\) in which each open leaf has been expanded
\emph{one step} according to the grammar. However, we only expand leaves which
are bound by a universal quantifier, that is, which represent possible subtrees
that could be unified with a universally quantified formula. For this reason, we
pass \(\Phi\) as an argument. We call the remaining, unbound grammar symbols
\emph{free nonterminals}.
For example, the XML constraint in Listing~\ref{main-lst:xml_constraint} from our paper \cite{steinhoefel.zeller-22*1}
does not restrict the instantiation of \(\bnfn{text}\) nonterminals. Thus,
\(\bnfn{text}\) is a free nonterminal which we will not expand with
\cref{rule:expand}. Instead, such nonterminals are instantiated to concrete
closed subtrees in a single step by
\cref{rule:finish_trivial,rule:finish_semantic}. In our implementation, we use a
standard coverage-based fuzzer to that end. Thus, we avoid producing many
strings which only differ, e.g., in identifier names or text passages within XML
tags.

\cref{rule:elim_smt_sat,rule:elim_sempred_sat} eliminate \emph{satisfiable} SMT
or semantic predicate formulas by querying \(\sinterpret\) or \(\pinterpret\)
(there is no transition for unsatisfiable or ``not ready'' formulas). The
transition result consists of one instantiation per returned assignment~\(\varassgn\). 

The only remaining constraints---in satisfiable constraint sets--- are
existential formulas, and semantic predicate formulas that are not yet ready to
provide a solution.
Existential formulas can be matched just like universal ones; but instead of
returning one result with all matches,
\cref{rule:match_exists,rule:match_exists_mexpr} return a \emph{set} of
solutions with one match each.

In addition to matching, we provide two rules
\cref{rule:elim_exists,rule:elim_exists_mexpr} to eliminate existential formulas
using \emph{tree insertion}. Note that, as exception to the general principle
that the rules in the \gls{cdtts} are mutually exclusive, we can apply
\cref{rule:match_exists,rule:match_exists_mexpr} \emph{and}
\cref{rule:elim_exists,rule:elim_exists_mexpr} wherever possible.
The insertion routine \(\trinsert{\mathit{newTree},t}\) guarantees that all
returned results contain all nodes from the original tree \(t\) as well as the
complete tree \(\mathit{newTree}\). Nevertheless, tree insertion is an
aggressive operation that may violate constraints that were satisfied before.
For this reason, we have to add the original constraint
\(\mathit{\Phi_\mathit{orig}}\), from which we started solving, again to the
constraint set; if the tree insertion did not violate structural constraints,
the original constraint can usually be quickly eliminated. However, tree
insertion can also entail the necessity of subsequent tree insertions, e.g., if
a new identifier was added that needs to be declared. Our implemented insertion
routine prioritizes structurally simple solutions, for which this is usually not
necessary. In the appendix, we provide details on tree insertion.

Finally, \cref{rule:finish_trivial,rule:finish_semantic} ``finish off'' the
remaining open derivation trees by replacing all open leaves with suitable
concrete subtrees. In the case of \cref{rule:finish_trivial}, this yields a
decisive result of the \gls{cdtts}. \cref{rule:finish_semantic} addresses
residual ``not ready'' semantic predicate formulas. We compute the represented
closed subtrees such that all information for evaluating the semantic predicates
is present. After this step, \cref{rule:elim_sempred_sat} must be applicable.

In the appendix, we argue for the correctness of the subsequent precision and
completeness theorems.

\begin{theorem}{(Precision)}
  The \isla{} \gls{cdtts} in \cref{fig:isla_cdtts} is globally precise.
  \label{thm:precision}
\end{theorem}

\begin{proof}[Proof Sketch.]
  By \cref{lem:local_implies_global_precision}, we prove global precision by
  showing that each individual transition rule is \emph{locally} precise, i.e.,
  that the produced states do not represent derivation trees that were not
  originally in the semantics of the inputs \gls{cdt}.

  \cref{rule:inv} is precise since conversion to \acl{dnf} is
  equivalence-preserving. Elimination of structural predicates
  (\cref{rule:elim_struct_true}) is trivially precise (removing a true element
  from a conjunction does not change the semantics).

  \cref{rule:elim_forall,rule:elim_forall_mexpr} are precise because a universal
  quantifier that does not match any tree element evaluates to true according to
  Definition~\ref{main-def:isla_validation} in our paper~\cite{steinhoefel.zeller-22*1}, and we only remove it if we
  can be sure that no possible extension of an open tree will ever match the
  quantifier. If a match is already in the index set, we can be sure that it
  already has been considered due to the definition of
  \cref{rule:match_forall,rule:match_forall_mexpr}, which are the only rules
  ever adding to that set.

  \cref{rule:match_forall,rule:match_forall_mexpr} are precise because we only
  \emph{add} the matching instantiations of the inner formula to the
  (conjunctive) constraint set. 

  Tree expansion (\cref{rule:expand}) is precise since by considering \emph{more
  concrete} trees, the set of concrete trees represented by the input \gls{cdt}
  is only ever \emph{decreased} in the outputs (cf.\
  \cref{def:semantics_of_open_derivation_trees}).

  % For the elimination of SMT formulas (\cref{rule:elim_smt_sat}), it is crucial
  % that \emph{no universal constraints} remain in the formula.

  The elimination of SMT formulas (\cref{rule:elim_smt_sat}) is precise since
  their semantics is defined via the interpretation function \(\sinterpret\),
  which we query to produce valid output states. The same holds for
  \cref{rule:elim_sempred_sat} for semantic predicates.

  Existential quantifier matching
  (\cref{rule:match_exists,rule:match_exists_mexpr}) is precise since it
  conforms to Definition~\ref{main-def:isla_validation} in our paper~\cite{steinhoefel.zeller-22*1} inasmuch it creates
  one instantiated \gls{cdt} for each match in the input \gls{cdt}. The original
  formula is removed from these results, but the instantiation retained.

  The tree insertion rules (\cref{rule:elim_exists,rule:elim_exists_mexpr}) (the
  most complicated ones in our system due to the complexity of tree tree
  insertion itself) are trivially to prove, because we add the additional
  constraint again to the constraint set.

  Finally, \cref{rule:finish_trivial,rule:finish_semantic} consider more
  concrete trees and are therefore precise for the same reasons as
  \cref{rule:expand}.
\end{proof}

\begin{theorem}{(Completeness)}
  The \isla{} \gls{cdtts} in \cref{fig:isla_cdtts} is globally complete.
  \label{thm:completeness}
\end{theorem}

\begin{proof}[Proof Sketch.]
  To prove the global completeness of our system, we have to show that the
  semantics of each input \gls{cdt} is contained in the semantics of all
  reachable \glspl{cdt} with empty constraint set.
  We reduce this problem as follows. First, we show \emph{local} completeness,
  i.e., that no information is lost by applying any transition rule of our
  \gls{cdtts}.  Second, we argue that for each \emph{valid} \gls{cdt}, there is
  an applicable rule in the \gls{cdtts}. Third, we argue that for each input
  \gls{cdt}, there is \emph{one} output \gls{cdt} which is \emph{closer} to a
  state with empty constraint set in the \gls{cdtts} than the inputs.
  From this, we conclude global completeness as follows: Since for each valid
  state, there is a transition step from which get closer to an output state
  with empty constraint set, this also holds for each valid state produced by
  this step. By additionally requiring that the individual steps do not lose
  information, we conclude global completeness.

  We argue for the local completeness of a chosen set of \gls{cdtts} rules.

  The expansion and finishing rules are locally complete if \emph{all} expansions
  are considered. This is the case in our \gls{cdtts}, although our actual
  implementation can only ever consider a finite set of solutions.

  The same holds for solving SMT formulas. Note that if we only consider a
  finite solution set as in our prototype, it is crucial that \emph{there remain
  no universal quantifiers} in the constraint set. Otherwise, we could obtain
  instantiations that conflict with formulas obtained from later quantifier
  instantiation. This is not a problem in the theoretic framework, though, since
  there we consider \emph{all possible} solutions, of which at least some will
  not conflict with atoms nested in remaining universal quantifiers.

  Tree insertion, which is easy to show precise, is more problematic to show
  locally complete, since we add the original constraint set. However, since we
  consider \emph{all} possible insertions, there have to be some satisfying that
  constraint, since the input \gls{cdt} is valid.

  Since we defined one rule for each syntactic construct in \isla{}, there is
  one rule for each valid input state. \cref{rule:expand}, for example, only
  expands nonterminals with potential concrete subtrees matching existing
  existential quantifiers; for all other nonterminals, the finishing
  rules will be applicable. 

  The general measure to show that each transition produces a state that is
  closer to an empty constraint set is the size of the constraint set together
  with the nesting depth of contained quantifiers. If either of these measures
  decreases in each step, we eventually reach an empty constraint set.
  That we get closer to an empty constraint set is clear to see for all
  elimination rules. In case of the matching rules, we reduce the complexity of
  the constraint set by peeling off the outer quantifier. Again, tree insertion
  is most problematic: It peels of the existential quantifier, but adds the
  original constraint set. Here, it is important to see that there are some
  insertions for which we can remove the existential constraint we eliminated by
  tree insertion by \emph{matching}, and that we thus do not have to keep
  re-inserting.
\end{proof}

We explain the main building blocks used in our \isla{} \gls{cdtts}
(\cref{fig:isla_cdtts}) in more detail. 

\paragraph{Tracking Instantiations}%
\label{par:tracking_instantiations}

Our \gls{cdtts} stepwise expands open trees and checks if existing universal
quantifiers match the expanded tree. Expansion does not eliminate a universal
quantifier, since it might apply to not yet generated subtrees. To avoid
endlessly instantiating universal quantifiers with the \emph{same} trees, we
track already performed universal quantifier instantiations. To that end, we
augment \glspl{cdt} with an \emph{index set} \(I\) consisting of pairs of
universal formulas and trees with which they already have been unified; we write
\(\condtreeidx{\Phi}{I}{t}\) for the enhanced structures.

\paragraph{Invariant}%
\label{par:invariant}

We maintain the invariant that all formulas \(\phi\in\Phi\) in \glspl{cdtts}
\(\condtree{\Phi}{t}\) are in \gls{nnf}, i.e., negations only occur directly
before predicate formulas and within SMT expressions, and are free of
conjunctions and disjunctions (on top level; they are allowed inside of
quantifiers and within SMT formulas). The function
\(\establishinvf:\fml\rightarrow\ps{\fml}\) first converts its argument into
\gls{nnf} by pushing negations inside (e.g.,
\(\islanot{\islaexists[\mathit{type}]{v}{w}{\phi}}\) gets
\(\islaforall[\mathit{type}]{v}{w}{\islanot{\phi}}\), and, for \(\psi\in\boolterms{V}\),
\(\islanot{\psi}\) gets \(\cd{(not $\psi$)}\in\boolterms{V}\)). Then, it
converts the result into \acl{dnf} by applying distributivity laws, which yields
a \emph{set} of disjunction-free formulas in \gls{nnf}. Finally, it splits all
top-level conjunctions outside SMT expressions in the result set into multiple
formulas.

\paragraph{SMT Models}%
\label{par:smt_models}

In \cref{fig:isla_cdtts}, we apply the interpretation \(\sinterpret\) for SMT
expressions to Boolean terms \(\boolterms{V}\) with a non-empty variable set
\(V\), i.e., the evaluated expressions may contain uninterpreted String
constants. In this case, the SMT solver will \emph{either} return \(\semfalse\)
in case of an unsatisfiable constraint (or time out, which we interpret as
\(\semfalse\)), \emph{or} an assignment \(\varassgn\) (a \emph{model}).
Since we can call the solver repeatedly and ask for different solutions (by
adding the negated previous solutions as assumptions), we assume that we get a
\emph{set} of assignments of tree identifiers to new subtrees from
\(\sinterpret\).

\paragraph{Semantic Predicates}%
\label{par:semantic_predicates}

We divide the set \(\psym\) of predicate symbols into two disjoint sets
\(\psymstruct\) and \(\psymsem\) of \emph{structural} and \emph{semantic}
predicates.
Structural predicates address structural constraints, such as \emph{before} or
\emph{within}. They evaluate to \(\semtrue\) or \(\semfalse\). \emph{Semantic}
predicates formalize more complex constraints, such as specific checksum
implementations. In addition to \(\semtrue\) or \(\semfalse\), semantic
predicate formulas may evaluate to a set of assignments, as in the case of
satisfiable SMT expressions, or to the special value ``not ready'' (denoted by
\(\notready\)).  Intuitively, an evaluation results in \(\semtrue\)
(\(\semfalse\)) if all of (not any of) the concrete derivation trees represented
by an abstract tree satisfy the predicate.  A set of assignments is returned if
there are reasonable ``fixes'' of the tree (e.g., all elements relevant for a
checksum computation are determined, such that the checksum can be computed by
the predicate). One may obtain \(\notready\) if the constrained tree lacks
sufficient information for such a computation; for instance, we cannot compute a
checksum if the summarized fields are still abstract.

In contrast to all other constraint types, the \emph{order of semantic predicate
formulas within a conjunction matters} (we use ordered sets in the
implementation of our \glspl{cdt}). The reason is that each semantic predicate
comes with its own, atomic solver. Consider, for example, a binary format which
requires a semantic predicate for the computation of a data field (e.g.,
requiring a specific compression algorithm) and another one for a checksum which
also includes the data field. Then, one must \emph{first} compute the value of
the data field, and \emph{then} the value of the checksum. Changing this order
would result in an invalid checksum. Since SMT formulas are composable, we
recommend using semantic predicates only if the necessary computation can either
not be expressed in SMT-LIB, or the solver frequently times out when searching
for solutions.

\paragraph{Tree Insertion}%
\label{par:tree_insertion}

Existential constraints can occasionally be solved by matching them against the
indicated subtree, similarly to universal quantifiers. In general, though, we
have to manipulate the tree to enforce the existence of the formalized
structure. If a successful match is not possible, we therefore constructively
\emph{insert} a new tree into the existing one.
The function \(\maketree{v}\) creates a new derivation tree consisting of a
single root node of type \(\vartype{v}\). When passing it a match expression
\(\mathit{mexpr}\) as additional argument, it creates a minimal open tree rooted
in a node of type \(\nlabel{v}\) and matching \(\mathit{mexpr}\).
The function \(\trinsert{t'}{t}\) tries to inserts the tree \(t'\) into \(t\).
Whether this is possible entirely depends on \(t\), \(t'\) and the grammar. In
the simplest case, \(t\) has an open leaf from which the nonterminal
\(\nlabel{t'}\) is reachable. Then, we create a suitable tree connecting the
leaf and the root of \(t'\) and glue these components together.

If this is not possible, we attempt to exploit recursive definitions in the
grammar. Consider, for example, a partial XML document according to the grammar
in Figure~\ref{main-fig:xml-bnf} in our paper~\cite{steinhoefel.zeller-22*1} and the constraint
\(\islaexists[\bnfn{xml\text{-}open\text{-}tag}]{\cd{optag}}{\cd{tree}}{\cd{(=
optag "<a>")}}\), where and \(\cd{tree}\) points to a node with root of type
\(\bnfn{xml\text{-}tree}\).  If there is some opening tag of form \cd{<a>} in
\(\mathit{tree}\), we can eliminate the constraint. Otherwise, we observe that
the nonterminal \(\bnfn{xml\text{-}tree}\) is \emph{reachable from itself} in
the grammar graph.
Thus, we can replace an existing \(\bnfn{xml\text{-}tree}\) node in
\(\mathit{tree}\) by a number of possible alternatives, comprising \(\cd{<a>}
\bnfs \bnfn{xml\text{-}tree} \allowbreak \bnfn{xml\text{-}close\text{-}tag}\),
which allows to insert both the already existing \(\bnfn{xml\text{-}tree}\) and
the new opening tag \cd{<a>} into the expanded result.  

\paragraph{Cost Function}

The choice of the right cost function is crucial for the performance of the
solver, both in terms of generation speed (number of outputs per time) and
output diversity (e.g., creation of deep nestings in the case of XML, or
coverage of combinations of language constructs in the case of C).

Our cost function computes the weighted geometric mean of \emph{cost factors}
\(\mathit{cf}_i\) and corresponding \emph{weights} \(w_i\) as
\[\mathit{cost}=\left(\prod_{i=1,\dots,n}^{w_i\neq0}(\mathit{cf}_i+1)^{w_i}\right)^{\left(\sum_{i=1,\dots,n}^{w_i\neq0}w_i\right)^{-1}}-1\]

We filter out pairs of cost factors and weights where the weight is 0; in this
case, the corresponding cost factor is deactivated. Furthermore, we avoid the
case that the final cost value is 0 if one of the factors is 0 by incrementing
each factor by 1, and finally decrementing the result by 1 again.

We chose the following cost factors:

\begin{description}[leftmargin=5mm]
  \item[Tree closing cost.] We precompute, for each nonterminal in the grammar,
    an approximation of the instantiation effort of that nonterminal, roughly by
    instantiating it several times randomly with a fuzzer, and then summing up
    the sizes of the possible grammar expansion alternatives in the resulting
    tree. The closing cost for a derivation tree is defined as the sum of the
    costs of each nonterminal symbol in all open leaves of the tree.
  \item[Constraint cost.] Certain constraints are more expensive to solve than
    others. In particular, solving existential quantifiers by tree insertion is
    computationally costly. This cost factor assigns higher cost for constraints
    with existential and deeply nested quantifiers.
  %
  %\item[Vacuous penalty.] Universal constraints can be ``solved'' with trees not
  %  matching the constraint. This is called \emph{vacuous
  %  satisfaction}~\cite{beer.ben-david.ea-97}. By assigning a cost to such
  %  trees, we force the solver to consider more interesting solutions.
  %
  \item[Derivation depth penalty.] As the solver's queue fills up, it becomes
    more improbable for individual queue elements to be selected next. If we
    assign a cost to the derivation depth, it becomes more likely that the
    solver eventually comes back to partial solutions discovered earlier,
    avoiding starvation of such inputs.
  \item[k-path coverage.] When choosing between different partial trees, we
    generally want to generate those exercising more language features at once.
    The k-path coverage metric~\cite{havrikov.zeller-19} computes all paths of
    length k in a grammar and derivation tree; the proportion of such paths
    covered by a tree is then the coverage value. We penalize trees which cover
    only few k-paths. The concrete value of k is configurable; the default is 3.
  \item[Global k-path coverage.] For each final result produced by the solver, we
    record the covered k-paths and from then on prefer solutions covering
    \emph{additional} language features. Once all k-paths in a grammar have been
    covered, we erase the record.
\end{description}

The influence of these cost factors can be controlled by passing a tuple of
weights to the solver. We provide a reasonable default vector
(\((11,3,5,20,10)\)), but in certain cases, a problem-specific tuning might be
necessary to improve the performance. Our implementation provides an
evolutionary parameter tuning mechanism, which runs the solver with randomly
chosen weights, and then computes several generations of weight vectors using
crossover and mutation. The fitness value of a weight vector is determined by
the generation speed, a vacuity estimator, and a k-path-based coverage measure.

\setcounter{section}{\getrefnumber{main-sec:evaluation}}
\addtocounter{section}{-1}
\DeclareRobustCommand{\sectitle}{\nameref*{main-sec:evaluation}}
\section{\sectitle}

\StrBehind{\getrefnumber{main-sub:rq3_islearn}}{.}[\subseccnt]
\setcounter{subsection}{\subseccnt}
\addtocounter{subsection}{-1}
\DeclareRobustCommand{\sectitle}{\nameref*{main-sub:rq3_islearn}}
\subsection{\nameref*{main-sub:rq3_islearn}}

In the subsequent
\crefrange{lst:islearn_constraint_dot}{lst:islearn_constraint_icmp_echo_checksum},
we list the constraints that \islearn{} mined in our case study for our third
research question.

\begin{lstlisting}[%
  language={isla},%
  % float=tb,%
  caption={Constraint mined by \islearn{} for DOT},%
  label=lst:islearn_constraint_dot,%
  basicstyle=\small\ttfamily,%
]
((forall <graph_type> container in start:
    exists <DIGRAPH> elem in container:
      (= elem "digraph") or
 forall <edge_stmt> container_0 in start:
   exists <edgeop> elem_0 in container_0:
     (= elem_0 "--")) and
(forall <graph> container_1 in start:
   exists <GRAPH> elem_1 in container_1:
     (= elem_1 "graph") or
forall <edge_stmt> container_2 in start:
  exists <edgeop> elem_2 in container_2:
    (= elem_2 "->")))
\end{lstlisting}
 
\begin{lstlisting}[%
  language={isla},%
  % float=tb,%
  caption={Constraint mined by \islearn{} for Racket based on the XML
  \emph{def-use} pattern for prefixes in attributes},%
  label=lst:islearn_constraint_racket,%
  basicstyle=\small\ttfamily,%
]
(forall <expr> attribute=
    "<maybe_comments><MWSS>{<name> prefix_use}" in start:
  ((= prefix_use "sqrt") or
   (= prefix_use "string-append") or
   $\dots$ or
  exists <definition> outer_tag="(<MWSS>define<MWSS>(<MWSS><name>{<WSS_NAMES> cont_attribute}<MWSS>)<MWSS><expr><MWSS>)" 
      in start:
    (inside(attribute, outer_tag) and
    exists <NAME> def_attribute="{<NAME_CHARS> prefix_def}" in cont_attribute:
      (= prefix_use prefix_def))))
\end{lstlisting}
 
\begin{lstlisting}[%
  language={isla},%
  % float=tb,%
  caption={Racket constraint based on the XML \emph{def-use} pattern in addition
  to an extended reST \emph{def-use} pattern},%
  label=lst:islearn_constraint_racket_with_rest,%
  basicstyle=\small\ttfamily,%
]
(forall <expr> use_ctx="<maybe_comments><MWSS>(<MWSS>{<name> use}<wss_exprs><MWSS>)" in start:
  ((= use "sqrt") or
   (= use "string-append") or
   $\dots$ or
    exists <definition> def_ctx=
        "(<MWSS>define<MWSS>(<MWSS>{<name> def}<WSS_NAMES><MWSS>)<MWSS><expr><MWSS>)" in start:
      ((before(def_ctx, use_ctx) and
      (= use def))))) and
(forall <expr> attribute=
    "<maybe_comments><MWSS>{<name> prefix_use}" in start:
  ((= prefix_use "sqrt") or
   (= prefix_use "string-append") or
   $\dots$ or
  exists <definition> outer_tag="(<MWSS>define<MWSS>(<MWSS><name>{<WSS_NAMES> cont_attribute}<MWSS>)<MWSS><expr><MWSS>)" in start:
    (inside(attribute, outer_tag) and
    exists <NAME> def_attribute="{<NAME_CHARS> prefix_def}" in cont_attribute:
      (= prefix_use prefix_def))))
\end{lstlisting}
  
\begin{lstlisting}[%
  language={isla},%
  % float=tb,%
  caption={\islearn{} constraint for ICMP Echo type fields},%
  label=lst:islearn_constraint_icmp_echo,%
  basicstyle=\small\ttfamily,%
]
(forall <icmp_message> container in start:
   exists <type> elem in container:
     (= elem "00 ") or
forall <icmp_message> container_0 in start:
  exists <type> elem_0 in container_0:
    (= elem_0 "08 ")))
\end{lstlisting}

\begin{lstlisting}[%
  language={isla},%
  % float=tb,%
  caption={Constraint learned by \islearn{} for ICMP Echo after adding a
  predicate for Internet Checksums},%
  label=lst:islearn_constraint_icmp_echo_checksum,%
  basicstyle=\small\ttfamily,%
]
((forall <icmp_message> container in start:
    exists <type> elem in container:
      (= elem "00 ") or
 forall <icmp_message> container_0 in start:
   exists <type> elem_0 in container_0:
     (= elem_0 "08 ")))) and
forall <icmp_message> container in start:
  exists <checksum> checksum in container:
    internet_checksum(container, checksum)
\end{lstlisting}

\balance
\bibliography{bibliography}
\clearpage

\end{document}

%% file: acronyms.tex
\newacronym{cfg}{CFG}{Context-Free Grammar}
\newacronym{cdt}{CDT}{Conditioned Derivation Tree}
\newacronym{cdtts}{CDTTS}{CDT Transition System}
\newacronym{clp}{CLP}{Constraint Logic Programming}
\newacronym{csg}{CSG}{Context-Sensitive Grammar}
\newacronym{dnf}{DNF}{Disjunctive Normal Form}
\newacronym{hoag}{HOAG}{Higher-Order Attribute Grammar}
\newacronym{isla}{ISLa}{Input Specification Language}
\newacronym{nnf}{NNF}{Negation Normal Form}
\newacronym{pbt}{PBT}{Property-Based Testing}
\newacronym{rest}{reST}{reStructuredText}

\newcommand{\isla}{ISLa\xspace}
\newcommand{\coreisla}{Core-ISLa\xspace}
\newcommand{\islearn}{ISLearn\xspace}